\RequirePackage{silence}

\documentclass[letterpaper,11pt]{article}
\usepackage[in]{fullpage}
\usepackage[normalem]{ulem}
\usepackage{amsmath}
\usepackage{amssymb}
\usepackage{amsthm}
\usepackage{amsfonts}
\usepackage{thmtools, thm-restate}
\usepackage{xcolor}
\usepackage{xspace}
\usepackage{calc}
\usepackage{comment}
\usepackage{embedfile}
\usepackage{algorithm}
\usepackage[noend]{algpseudocode}
\usepackage{parskip}
\usepackage{multirow}
\usepackage{mathtools}
\usepackage[T1]{fontenc}
\usepackage{wrapfig}
\usepackage{bbm}
\usepackage{mfirstuc}
\usepackage{hyperref}
\hypersetup{
	colorlinks,
	citecolor=blue,
	filecolor=black,
	linkcolor=black,
	urlcolor=black,
	linktoc=all
}

\newif\ifdeletegray
\ifdeletegray
\excludecomment{gray}
\fi

\algnewcommand\algorithmicwp{\textbf{with probability}}%
\algdef{SE}[WP]{WithProb}{EndWithProb}[1]%
{\algorithmicwp\ #1\ \algorithmicdo}%
{}%
\algtext*{EndWithProb}%

\newcounter{savealgorithm}

\newtheorem{lemma}{Lemma}[section]

\newtheorem{observation}[lemma]{Observation}

\newtheorem*{claim*}{Claim}
\newtheorem*{corollary*}{Corollary}
\newtheorem*{observation*}{Observation}

\newtheorem{speculative-lemma}{Lemma (Speculative)}
\newtheorem{speculative-theorem}[speculative-lemma]{Theorem (Speculative)}
\theoremstyle{definition}\newtheorem{definition}[lemma]{Definition}
\theoremstyle{definition}\newtheorem{speculative-definition}[speculative-lemma]{Definition (Speculative)}
\newcommand{\accept}[0]{\textsc{accept}\xspace}
\newcommand{\reject}[0]{\textsc{reject}\xspace}

\newcommand{\abs}[1]{\left|{#1}\right|}

\newcommand{\floor}[1]{{\left\lfloor{#1}\right\rfloor}}
\newcommand{\ceil}[1]{{\left\lceil{#1}\right\rceil}}
\newcommand{\lfrac}[2]{\left. {#1} \middle/ {#2} \right.}

\newcommand{\E}{\mathop{{\rm E}\/}}

\newcommand{\eps}{\varepsilon}

\newcommand{\supp}{\mathrm{supp}}

\newcommand{\poly}{\mathrm{poly}}
\newcommand{\Ber}{\mathrm{Ber}}

\newcommand{\kl}{\ensuremath{\mathrm{KL}}}
\newcommand{\dkl}{\ensuremath{D_\kl}}
\newcommand{\dtv}{\ensuremath{d_\mathrm{TV}}}

\newcommand{\propname}[1]{\textsc{\capitalisewords{#1}}}

\title{Improved Bounds for High-Dimensional Equivalence and Product Testing using Subcube Queries}
\author{Tomer Adar\thanks{Technion - Israel Institute of Technology, Israel. Email: \href{mailto:tomer-adar@campus.technion.ac.il}{tomer-adar@campus.technion.ac.il}.} \and Eldar Fischer\thanks{Technion - Israel Institute of Technology, Israel. Email: \href{mailto:eldar@cs.technion.ac.il}{eldar@cs.technion.ac.il}. Research supported by an Israel Science Foundation grant number 879/22.} \and Amit Levi\thanks{University of Haifa and Technion - Israel Institute of Technology, Israel. Email: \href{mailto:alevi@ds.haifa.ac.il}{alevi@ds.haifa.ac.il}.}}

\begin{document}

\begin{titlepage}
    \maketitle
    \thispagestyle{empty}
    \begin{abstract}
        We study property testing in the subcube conditional model introduced by Bhattacharyya and Chakraborty (2017). We obtain the first equivalence test for $n$-dimensional distributions that is quasi-linear in $n$, improving the previously known $\tilde{O}(n^2/\eps^2)$ query complexity bound to $\tilde{O}(n/\eps^2)$. We extend this result to general finite alphabets with logarithmic cost in the alphabet size.
        
        By exploiting the specific structure of the queries that we use (which are more restrictive than general subcube queries), we obtain a cubic improvement over the best known test for distributions over $\{1,\ldots,N\}$ under the interval querying model of Canonne, Ron and Servedio (2015), attaining a query complexity of $\tilde{O}((\log N)/\eps^2)$, which for fixed $\eps$ almost matches the known lower bound of $\Omega((\log N)/\log\log N)$. We also derive a product test for $n$-dimensional distributions with $\tilde{O}(n / \eps^2)$ queries, and provide an $\Omega(\sqrt{n} / \eps^2)$ lower bound for this property.
    \end{abstract}
\end{titlepage}

%\tableofcontents
%\thispagestyle{empty}

%\newpage
\pagenumbering{arabic}

\section{Introduction}
\label{sec:intro}
Property testing seeks to efficiently distinguish between objects that have some property and objects that are $\eps$-far from any object that has it, with respect to a predefined metric and a proximity parameter $\eps$. Property testing of functions, and in particular string testing, was initiated in \cite{rs1996,blr1993}. The term ``efficiently'' usually refers to a sublinear amount of resources for moderately-sized tasks, but for high-dimensional inputs practicality mandates to restrict this amount further to the poly-logarithmic scale.

The study of distribution testing was implicitly initiated in \cite{gr11} (originally as \cite{gr11asGR00}, motivated by a problem in graph testing), and formally defined as the \emph{sampling model} in \cite{batu-FRSW2000, batuFFKRW2001}. In this model, the algorithm gets access to a sequence of independent samples identically distributed from the input distribution, and then decides whether to accept or reject based on them. A major topic of investigation concerns testing that a distribution over $\{1,\ldots,n\}$ is the uniform one. A long string of results, starting with the original \cite{gr11} and culminating in \cite{paninski2008coincidence} (for most values of $\eps$; for extreme cases, it has been completed by \cite{ADK15} and \cite{DK16}), has reached the tight bound of $\Theta(\sqrt{n}/\eps^2)$.

A square-root lower bound on one of the most basic of properties is impractical in many real-world settings. For example, the square-root sample complexity of uniformity testing is impractical when the input ranges over $100$-bit binary strings, an example which is still smaller than real-world inputs. 
There are three approaches to overcome this problem: the first, on which we expand below, is moving to stronger query models. The second is narrowing the scope of the allowable inputs, such as \emph{product distributions} and \emph{Bayesian networks} \cite{baynets17}. The third is moving to a model where the metric is more lax (possibly coupled with an even weaker query model), such as Hellinger distance \cite{DKW18} and the Huge Object Model defined in \cite{gr22} that is concerned with the earth-mover distance metric.

We focus here on the first (and most common) approach to the scaling problem, that of considering a model with stronger queries. In the distribution testing setting, this usually focuses on \emph{conditional sampling}. Instead of drawing a sequence of independent unconditional samples, we allow the algorithm to choose a subset of the possible outcomes and draw a sample conditioned on belonging to it. The models within this paradigm differ in the subset conditions they allow.

The first investigation of such a model was the fully conditional model \cite{chakraborty2013power,crs15}. In this model, the algorithm can choose \emph{any} subset to condition on. This model is very powerful but not realistic. If we are able to restrict the input distribution to any subset, we probably already have access to an explicit description of the distribution and thus have no need to sample it.

Further studies consider restricted forms of the conditional model. For example, \cite{crs15} considers two additional models: the pair model and the interval model. In the pair model, the algorithm can still draw unconditional samples, and additionally it can draw conditional samples from any subset of two elements. Uniformity testing in the pair model requires $\tilde{\Theta}(1/\eps^2)$ samples. In the interval model, we test distributions over $\Omega = \{1,\ldots,N\}$ (for some $N$), and the algorithm can condition on interval sets, which are sets of the form $\{ a \le x \le b \}$. The lower bound for uniformity testing in the interval model is $\tilde{\Omega}(\log N)$. The best known upper bound is $\tilde{O}((\log N)^3 / \eps^2)$.

The subcube conditional model \cite{bc17} is motivated by database analysis. In this model, we test distributions over the set $\prod_{i=1}^n \Omega_i$, and the algorithm can query \emph{subcube subsets}, which are sets of the form $\prod_{i=1}^n A_i$ where $A_i \subseteq \Omega_i$ for every $1 \le i \le n$. While not being extremely restrictive, its queries correspond to selection by attribute values, which is common in practice. Some of the prior work refers to a weaker variant of this model, where each $A_i$ is either trivial ($A_i = \Omega_i$) or a singleton ($A_i=\{a_i\}$ for some $a_i\in\Omega_i$). In this paper we mainly deal with the strong model, where general $A_i\subseteq\Omega_i$ are allowed. In the most investigated setting, the binary setting where $\Omega_i=\{0,1\}$ for all $i$, the two models are the same.

Uniformity testing \cite{canonne2021random,chen2024uniformity} can be done in the weak model using $\tilde{O}(m^{21} \sqrt{n} / \eps^2)$ queries where $m = \max_i |\Omega_i|$, and requires at least $\Omega(\sqrt{mn} / \eps^2)$  queries \cite{bhattacharyya2021testing}.  Other properties studied under the subcube model include identity to some fixed distribution \cite{bc17,horrible} and having a probability density function supported on a low-dimensional subspace \cite{chen2021learning}.

Two related properties are of particular interest in distribution testing, the identity property and the equivalence property. In the identity property, a \emph{reference distribution} is given to the algorithm in advance, and the task is to check whether the input distribution is identical to it. The equivalence property (henceforth: \propname{equivalence}) has an input consisting of \emph{two} distributions, both of which accessible through the testing model, and the task is to check whether they are equal to each other.

In the subcube conditional model, since the input distribution is defined over a set of tuples, another natural property is that of being a product distribution (henceforth: \propname{product}, also known as independence testing in \cite{batuFFKRW2001, ADK15, DK16}). A distribution over tuples is called a product if its entries are independently distributed.

Our main upper bound result is a test for \propname{equivalence} that for the binary setting ($\Omega_i=\{0,1\}$) uses only $\tilde{O}(n/\epsilon^2)$ subcube queries, which improves on the previously known result of $\tilde{O}(n^2/\eps^2)$ from \cite{bc17}. Additionally our result uses only \emph{prefix queries} from one distribution and \emph{marginal prefix queries} from the other, which we define below and are rather restricted forms of subcube queries. One can think of prefix queries as the queries that can be made fast when the database is sorted according to a concatenation of its attributes in a pre-defined order (functioning as its primary key). Importantly, the use of restricted queries allows us to derive an improved test also for \propname{product}, and to generalize the test to general $\Omega_1,\ldots,\Omega_n$ with a logarithmic cost in the alphabet size.

The restricted form of our queries also allows us to tighten the previously known upper bound on equivalence testing (and through it the special case of uniformity testing) in the interval conditional model, obtaining an upper bound of $\tilde{O}(\log N / \eps^2)$ interval queries, which matches the lower bound for every fixed $\eps$ up to poly-double-logarithmic factors in $N$.

We complement our upper bound for \propname{product} with an $\Omega(\sqrt{n}/\eps^2)$ lower bound. The question of whether we can go below $O(n)$ for testing our properties (even for the binary setting and a fixed $\eps$) remains open.

\section{Organization of the paper}
\label{sec:orgainzation}
In Section \ref{sec:results} we summarize our contributions. After some preliminaries in Section \ref{sec:prelims}, we provide the core of our main proofs, followed by the more technical details. In Section \ref{sec:closeness-main-path} we provide the \propname{equivalence} testing upper bound in the binary setting, and a short proof of the corollary about interval queries. In Section \ref{sec:lbnd-product-main-path} we provide the \propname{product} testing lower bound.

% The technical part of the paper follows. In Section \ref{sec:technical-proofs-closeness} we prove the technical lemmas whose proofs were deferred from Section \ref{sec:closeness-main-path}.

Then we prove the theorems derived from this upper bound: in Section \ref{sec:closeness-omegas} we extend the \propname{equivalence} test to the non-binary setting, and in Section \ref{sec:ubnd-product} we derive a test for \propname{product}.
% In Subsection \ref{sec:lbnd-product-main-path:subsec:technical-proofs-product} we prove the technical lemmas from Section \ref{sec:lbnd-product-main-path}.

All upper bound proofs implicitly construct their algorithms. For reference, explicit representations of the binary setting algorithms for \propname{equivalence} and \propname{product} are given in the appendix.

\section{Our results}
\label{sec:results}
We improve on the previously known result of $\tilde{O}(n^2/\eps^2)$ queries for equivalence testing of two distributions over $\{0,1\}^n$ \cite{bc17}. Our methods can also supersede the $\tilde{O}(n/\eps)$ algorithm of \cite{horrible} which tests identity with a distribution given in advance that belongs to a very restricted class of inputs (their parameter refers to KL-divergence, which indeed incurs a quadratic gap when converted to total-variation distance). We provide more details on the latter below (before Lemma \ref{lemma:closeness-linearithmic-fd-soundness-if-nice}).

\begin{restatable*}{theorem}{thZclosenessZlinearithmic}\label{th:closeness-linearithmic}
    Let $\tau$, $\mu$ be two distributions over $\{0,1\}^n$, where $\tau$ is accessible through the prefix oracle access and $\mu$ is accessible through the marginal prefix oracle access. For every $\eps > 0$ we can distinguish between $\tau = \mu$ and $\dtv(\tau,\mu) > \eps$ using $\tilde{O}(n / \eps^2)$ queries.
\end{restatable*}

The prefix queries in the statement above are a special case of subcube queries, but they can also be seen as interval queries, allowing us to prove the following corollary:

\begin{restatable*}{corollary}{corZintervalZlinearZtilde} \label{cor:interval-linear-tilde}
    Let $\tau$, $\mu$ be two distributions over $\{1,\ldots,N\}$, both accessible through the interval oracle. Then we can distinguish between $\tau = \mu$ and $\dtv(\tau, \mu) > \eps$ using $\tilde{O}((\log N) / \eps^2)$ queries.
\end{restatable*}

We show a lower bound for testing a distribution $\mu$ over $\{0,1\}^n$ for being a product. To the best of our knowledge, it is the first lower bound for product testing  in the binary setting. Our construction is similar to the lower bound for uniformity testing of \cite{baynets17}.

\begin{restatable*}{theorem}{thZlbndZindependence} \label{th:lbnd-independence}
    Every $\eps$-test for \propname{product} must make at least $\Omega\left(\sqrt{n} / \eps^2\right)$ subcube queries.
\end{restatable*}

We generalize our upper bound for equivalence testing to strings of size $n$ over larger alphabets. Our result is incomparable with the previously known result of $\tilde{O}(\frac{n^5}{\eps^5} \log \log |\Omega|)$ \cite{bc17}. Note that the cited result refers to strings over a single alphabet (that is, $\Omega^n$), whereas our result refers to strings over mixed alphabets (that is, $\prod_{i=1}^n \Omega_i$). Additionally, our result is more efficient when $|\Omega|$ is not very large with respect to $n$.

\begin{restatable*}{theorem}{thZclosenessZomegas} \label{th:closeness-omegas}
    Let $\mu$ and $\tau$ be two distributions over $\prod_{i=1}^n \Omega_i$, where $\Omega_1,\ldots,\Omega_n$ are all finite. If $\mu$ is accessible through the marginal prefix oracle and $\tau$ is accessible through the prefix oracle, then we can distinguish between $\tau = \mu$ and $\dtv(\tau, \mu) > \eps$ using $\tilde{O}(\sum_{i=1}^n \log_2 |\Omega_i| / \eps^2)$ queries.
\end{restatable*}

We apply the same generalization for product testing as well. In the binary setting, it also improves on the previously known result of $\tilde{O}(n^2 / \eps^2)$ \cite{bc17}.

\begin{restatable*}{theorem}{thZubndZproduct} \label{th:ubnd-product}
    Let $\mu$ be a distribution over $\prod_{i=1}^n \Omega_i$. For every $0 < \eps < 1$, we can distinguish between the case where $\mu$ is a product distribution and the case where it is $\eps$-far from every product distribution at the cost of $\tilde{O}(\sum_{i=1}^n \log |\Omega_i| / \eps^2)$ subcube queries. Moreover, if $|\Omega_i| = 2$ for every $1 \le i \le n$, then all these queries are prefix queries.
\end{restatable*}

\section{Preliminaries}
\label{sec:prelims}

For brevity, for $m\in\mathbb{N}$ we let $[m]$ denote the set $\{1,\ldots,m\}$.

\begin{definition}[Bernoulli distribution]
    Let $0 \le p \le 1$. The \emph{Bernoulli distribution with parameter $p$}, denoted by $\Ber(p)$, is the distribution over $\{0,1\}$ whose probability to draw $1$ is $p$.
\end{definition}

\begin{definition}[Common statistical divergence measures]
    Let $\mu$ and $\tau$ be two distributions over a finite set $\Omega$. We use two well known divergence measures, namely the total-variation distance $\dtv(\mu, \tau) = \frac{1}{2} \sum_{x \in \Omega} \abs{\mu(x) - \tau(x)} = \max_{E \subseteq \Omega} \abs{\mu(E) - \tau(E)}$ and the KL-divergence $\dkl(\mu, \tau) = \E_{x \sim \mu}\left[\log_2 \frac{\mu(x)}{\tau(x)} \right]$. Note that this notation is only well-defined where $\Pr_\mu[B] > 0$.
\end{definition}

\begin{definition}[Conditional distribution]
    Let $\mu$ be a distribution over $\Omega$, and let $B \subseteq \Omega$ be an event. The corresponding \emph{conditional distribution} is denoted by $\mu|^B$ and is defined by $\mu|^B(x)=0$ for $x\notin B$ and $\mu|^B(x)=\mu(x)/\Pr_{\mu}[B]$ for $x\in B$.
\end{definition}

\begin{definition}[Index-restricted distribution]
    Let $\mu$ be a distribution over $\prod_{i=1}^n\Omega_i$, and let $I \subseteq [n]$ be a set of indices. We use $\mu|_I$ to denote the distribution that draws $x \sim \mu$, and returns the restricted string $x|_I \in \prod_{i\in I}\Omega_i$.
\end{definition}

Note that the correct way to parse the overloaded notation $\mu|^B_I$ is ``$(\mu|^B)|_I$'', that is, we first apply the condition and then restrict the indices.

Here we define the query model for distributions over binary strings. The definitions for general alphabets appear in Section \ref{sec:closeness-omegas}.

\begin{definition}[Subcube oracle access]
    Let $\mu$ be an unknown distribution over $\{0,1\}^n$. The \emph{subcube oracle} has a set $I \subseteq [n]$ and a string $w \in \{0,1\}^I$ as input, and its output distributes as $\mu|^{x_I = w}$. If $\Pr_\mu[x_I = w] = 0$, then the oracle indicates an error. Note that the answers of the oracle are fully independent of the answers that were given to previous queries.
\end{definition}

We note that the ``error behavior'' in the above definition does not really affect our results, since all our upper bounds use only restrictions to guaranteed positive probability outcomes, while our lower bounds use distributions that have no zero-probability elements.

\begin{definition}[Prefix oracle access]
    Let $\mu$ be an unknown distribution over $\{0,1\}^n$. The \emph{prefix oracle} is a restricted case of the subcube oracle, where $I = [k]$ for some $0 \le k \le n - 1$.
\end{definition}

\begin{definition}[Marginal subcube oracle access]
    Let $\mu$ be an unknown distribution over $\{0,1\}^n$. The \emph{marginal subcube oracle} has a set $I \subseteq [n]$, an index $i \in [n] \setminus I$ and a string $w \in \{0,1\}^I$ as input, and its output is a single bit that distributes as $\mu|_i^{x_I = w}$. If $\Pr_\mu[x_I = w] = 0$, then the oracle indicates an error.
\end{definition}

\begin{definition}[Marginal prefix oracle access]
    Let $\mu$ be an unknown distribution over $\{0,1\}^n$. The \emph{marginal prefix oracle} is a restricted case of the marginal subcube oracle, where $I \!=\! [i-1]$.
\end{definition}

We now define the interval querying model for which we derive a new bound.

\begin{definition}[Interval oracle access]
    Let $\mu$ be an unknown distribution over $[N]$. The \emph{interval oracle} has two elements $1\leq a\leq b\leq N$ as input, and its output distributes as $\mu|^{\{a,\ldots,b\}}$.
\end{definition}

We next define our properties for distributions over general alphabet strings, and what it means to test for a property.

\begin{definition}[\propname{equivalence}]
    A pair of distributions $\mu$, $\tau$ over $\prod_{i=1}^n\Omega_i$ belongs to the equivalence property if $\mu = \tau$. The distance of a given pair $(\mu,\tau)$ from \propname{equivalence} is explicitly defined as $\dtv(\mu,\tau)$ instead of the natural $\inf_\nu (\dtv(\mu,\nu)+ \dtv(\tau,\nu))$. The two quantities are easily seen to be identical using the triangle inequality.
\end{definition}

\begin{definition}[\propname{product}]
    A distribution $\mu$ over $\prod_{i=1}^n\Omega_i$ is called a \emph{product distribution} if there exist distributions $\mu_1,\ldots,\mu_n$ over $\Omega_1,\ldots,\Omega_n$ respectively, for which $\mu \sim \mu_1 \times \cdots \times \mu_n$. We denote by \propname{Product} the set of all product distributions over $\prod_{i=1}^n\Omega_i$.
\end{definition}

\begin{definition}[$\eps$-test]
    For some property $\mathcal{P}$ of distributions (which is a set of distributions) and $\eps > 0$, we say that an algorithm $\mathcal{A}$ is an $\eps$-test for $\mathcal{P}$ if:
    \begin{itemize}
        \item For every input distribution $\mu \in \mathcal{P}$, $\mathcal{A}$ accepts with probability at least $\frac{2}{3}$.
        \item For every input distribution $\mu$ for which $\dtv(\mu,\nu) > \eps$ for all $\nu \in \mathcal{P}$, $\mathcal{A}$ rejects with probability at least $\frac{2}{3}$.
    \end{itemize}
\end{definition}

\section{Linear sample-complexity algorithm for \propname{equivalence}}
\label{sec:closeness-main-path}
To construct our improvement on the algorithm of \cite{bc17}, we first define some templates for analyzing and comparing differences between distributions.

\begin{definition}[Single-bit divergence]
    We call $d : [0,1] \times [0,1] \to [0,\infty)$ a \emph{single-bit divergence} if:
    \begin{itemize}
        \item For every $p,q \in [0,1]$, $d(p,q) = 0$ if and only if $p=q$ (\emph{Positivity}).
        \item For every $p' \le p \le q \le q'$, $d(p,q) \le d(p', q')$ and $d(q,p) \le d(q', p')$ (\emph{Monotonicity}).
    \end{itemize}
\end{definition}
Three useful single-bit divergences are $\mathrm{TV}(p,q) = |p - q|$, $\kl(p,q) = p \log_2 \frac{p}{q} + (1-p)\log_2 \frac{1-p}{1-q}$ and $\chi^2(p,q) = \frac{(p-q)^2}{(p+q)(2 -(p + q))}$. Note that all three are derived from their probability-theoretic counterparts for $\Ber(p)$ and $\Ber(q)$. We have two motivations to prefer this symmetric form of $\chi^2$: first, the symmetry matches the idea of two unknown distributions, which is not the case in standard $\chi^2$-tests, and second, it is bounded by $1$, which makes the analysis more similar to total-variation distance than to KL-divergence. In this paper, we use the name ``chi-square'' (and the notation $\chi^2$) to indicate this symmetric form.

\begin{definition}[Slice-wise divergence]
    Let $d$ be a single-bit divergence. For every $n$ and two distributions $\tau$, $\mu$ over $\{0,1\}^n$, the \emph{slice-wise divergence} of $\tau$ and $\mu$ with respect to $d$ is:
    \[ \Delta_d(\tau, \mu) = \sum_{i=1}^n \E_{w \sim \tau} \left[d(\tau|_i^{x_{[i-1]} = w_{[i-1]}}(1), \mu|_i^{x_{[i-1]} = w_{[i-1]}}(1))\right] \]
    (note that the $d$-divergence is fed the probabilities of the two single-bit distributions to draw $1$).
\end{definition}

Recall the $\tilde{O}(n^5 / \poly(\eps))$ algorithm of \cite{bc17}, which is actually $\tilde{O}(n^2 / \eps^2)$ for the binary setting. As a motivation, the paper uses (and has a self-contained proof of) the following bound:
\[ \dtv(\tau, \mu) \le \sum_{i=1}^n \E_{w \sim \tau} \left[\dtv(\tau|_i^{x_{[i-1]} = w_{[i-1]}}, \mu|_i^{x_{[i-1]} = w_{[i-1]}})\right] = \Delta_\mathrm{TV}(\tau, \mu) \]
The algorithm then distinguishes between $\Delta_\mathrm{TV}(\tau, \mu) = 0$ (which always holds if $\tau = \mu$) and $\Delta_\mathrm{TV}(\tau, \mu) > \eps$ (which always holds if $\dtv(\tau, \mu) > \eps$).

To improve their test, inside the slice-wise divergence expression we substitute the single-bit total-variation distance with the single-bit $\chi^2$-distance, and analyze the more convenient $\Delta_{\chi^2}(\tau,\mu)$.

The following is immediate, and in fact holds for every slice-wise divergence:

\begin{observation}\label{obs:closeness-linearithmic-fd-completeness}
    Let $\tau$, $\mu$ be two distributions over $\{0,1\}^n$. If $\tau = \mu$ then $\Delta_{\chi^2}(\tau, \mu) = 0$.
\end{observation}

The following lemma, which we prove in Subsection \ref{sec:closeness-main-path:subsec:soundness}, provides the conversion from the total variation distance to the slice-wise chi-square divergence.

\begin{lemma}\label{lemma:closeness-linearithmic-fd-soundness}
    Let $\tau$, $\mu$ be two distributions over $\{0,1\}^n$. If $\tau \ne \mu$, then
    $\Delta_{\chi^2}(\tau, \mu) \ge \frac{(\dtv(\tau,\mu))^2}{24 \log ({2n}/\dtv(\tau,\mu))}$.
\end{lemma}

The following lemma, which we prove in Subsection \ref{sec:closeness-main-path:subsec:chi-square-test}, states that we can distinguish between single bit distributions using linearly many samples with respect to the inverse of their $\chi^2$-divergence.

\begin{restatable}{lemma}{lemmaZsingleZbitZchiZsquareZtest}
    \label{lemma:single-bit-chi-square-test}
    Let $p,q \in [0,1]$ be two probabilities. Given unconditional sampling access to $\Ber(p)$ and $\Ber(q)$, we can distinguish, with probability $\frac{2}{3}$, between the case where $p=q$ and the case where $\chi^2(p, q) > \eps$, at the cost of $O(1 / \eps)$ samples from each of them.
\end{restatable}

We use a common variant of the Levin's work balance method:

\begin{lemma}[\cite{levin1985one}, optimization exercise 8.4 in \cite{goldreich2017introduction}] \label{lemma:levin}
    Let $X$ be a non-negative random variable that is bounded by $1$. Assume that there exists some random variable $Y$ such that for every $y \in \supp(Y)$ and every $\eps' > 0$, we can distinguish between $\E[X|Y=y] = 0$ and $\E[X|Y=y] > \eps'$ using some black-box algorithm whose resource cost is $O(1/\eps')$. Also, assume that we can draw independent unconditional samples from $Y$ at resource cost $O(1)$ per sample. Then we can distinguish between $\E[X] = 0$ and $\E[X] > \eps$ at a total resource cost of $O(\eps^{-1} \log (1/\eps))$.
\end{lemma}

Using the above we can efficiently detect a large slicewise $\chi^2$ divergence between two distributions.

\begin{restatable}{lemma}{lemmaZclosenessZlinearitmicZDeltaZtest}\label{lemma:closeness-linearitmic-Delta-test}
    Let $\tau$, $\mu$ be two distributions over $\{0,1\}^n$. Then for every $\rho > 0$, we can distinguish between $\Delta_{\chi^2}(\tau, \mu) = 0$ and $\Delta_{\chi^2}(\tau, \mu) > \rho$ using $O(\frac{n}{\rho} \log \frac{n}{\rho})$ prefix queries.
\end{restatable}

\begin{proof}
    We normalize the divergence:
    \begin{eqnarray*}
        \frac{1}{n}\Delta_{\chi^2}(\tau,\mu)
        &=& \frac{1}{n}\sum_{i=1}^n \E_{w \sim \tau} \left[\chi^2(\tau|_i^{x_{[i-1]} = w_{[i-1]}}(1), \mu|_i^{x_{[i-1]} = w_{[i-1]}}(1)) \right] \\
        &=& \E_{\substack{w \sim \tau \\ i \sim [n]}} \left[\chi^2(\tau|_i^{x_{[i-1]} = w_{[i-1]}}(1), \mu|_i^{x_{[i-1]} = w_{[i-1]}}(1)) \right]
    \end{eqnarray*}

    Then we apply Levin's method, Lemma \ref{lemma:levin}, with the following input. $Y$ is a random variable that receives a value $(i,w)$ where $i$ is uniformly drawn from $[n]$ and $w$ is drawn from $\tau$ (independently of $i$). The parameter $\eps$ is set to $\rho/n$, and the random variable $X$ is defined as a function of $Y=(i,w)$ (meaning that it is constant when conditioned on a specific value of $Y$) by
    \[X(i,w)=\chi^2\left(\mu|_i^{x_{[i-1]} = w_{[i-1]}}(1), \tau|_i^{x_{[i-1]} = w_{[i-1]}}(1)\right)\]
    The distinction between the cases $X(Y)=0$ and $X(Y)>\eps'$ is performed using $O(1/\eps')$ many queries (our resource cost) through the single-bit $\chi^2$-test of Lemma \ref{lemma:single-bit-chi-square-test}.
\end{proof}

Our main theorem follows immediately from the above statements:

\thZclosenessZlinearithmic

\begin{proof}
  Apply Lemma \ref{lemma:closeness-linearitmic-Delta-test} using $\rho = \frac{\eps^2}{24 \log (2n / \eps)}$. Completeness follows from Observation \ref{obs:closeness-linearithmic-fd-completeness} and soundness follows from Lemma \ref{lemma:closeness-linearithmic-fd-soundness}.
\end{proof}

\corZintervalZlinearZtilde

\begin{proof}
    Without loss of generality, assume that $N=2^\ell$ for some integer $\ell$ (otherwise we just pad the two distributions with zero-probability elements). For every $i \ge 0$, let $\mathrm{bin}_i : \{0,\ldots,2^i - 1\} \to \{0,1\}^i$ and $\mathrm{unbin}_i : \{0,1\}^i \to \{0,\ldots,2^i-1\}$ be the mappings between small integers and their representation as $i$-bit strings. That is, $\mathrm{unbin}_i (x_1,\ldots,x_i)=\sum_{j=1}^i2^{i-j}x_i$ and $\mathrm{bin}_i=(\mathrm{unbin}_i)^{-1}$.
    
    To apply our equivalence algorithm on distributions over $[N]$, every $t \in [N]$ is interpreted as an $\ell$-bit string (using the map $t \to \mathrm{bin}_\ell(t-1)$). Then every prefix query $\tau|^{x|_{[i-1]} = w}$ (respectively $\mu|^{x|_{[i-1]} = w}$) is simulated using the interval query $\tau|^{\{a,\ldots,b\}}$ (respectively $\mu|^{\{a,\ldots,b\}}$), where $a = 2^{N-i} \mathrm{unbin}_{i-1}(w) + 1$ and $b = 2^{N-i} (\mathrm{unbin}_{i-1}(w) + 1)$. Since the simulated algorithm uses $\tilde{O}(\ell / \eps^2)$ prefix queries, the simulation uses $\tilde{O}((\log N) / \eps^2)$ interval queries.
\end{proof}

\subsection{Proof of Lemma \ref{lemma:closeness-linearithmic-fd-soundness}}
\label{sec:closeness-main-path:subsec:soundness}
The proof works by comparing the TV-distance and the slice-wise chi-square divergence with the KL-divergence, which is equal to its slice-wise version as per the following lemmas.\begin{lemma}[Basic chain rule for $\dkl$, folklore]
    Let $\mu$ and $\tau$ be distributions over $\Omega$. Then for every random variable $X : \Omega\to R$, we have the equality $\dkl(\mu,\tau) = \dkl(X(\mu),X(\tau)) + \E_{x \sim X(\tau)}\left[\dkl\left(\mu|^{X = x}, \tau|^{X=x}\right)\right]$.
\end{lemma}

\begin{lemma}[Chain rule for $\dkl$, repeated form] \label{lemma:chain-rule-dkl-01n}
    Let $\mu$ and $\tau$ be distributions over $\{0,1\}^n$. Then $\dkl(\mu,\tau) = \sum_{i=1}^n \E_{w \sim \mu}\left[ \dkl(\mu|_i^{x_{[i-1]} = w_{[i-1]}}, \tau|_i^{x_{[i-1]} = w_{[i-1]}}) \right] = \Delta_\kl(\mu,\tau)$.
\end{lemma}

\begin{proof}
    We obtain this result by repeatedly applying the basic chain rule using induction.
\end{proof}

We recall two folklore bounds for the KL-divergence of Bernoulli distributions.

\begin{lemma}[Well known]
    For $p,q \in [0,1]$, $\kl(p, q) \le \frac{(p-q)^2}{q(1-q)}$. More formally, in $(0,1) \times (0,1)$, the ratio between these expressions is a non-negative continuous function that is bounded by $1$.
\end{lemma}
\begin{lemma}[Direct corollary] \label{lemma:KL-bound-aq-q}
    Let $p,q \in [0,1]$. If $p = aq$ for some real $a$, then
    \[\kl(aq, q) \le \frac{(a-1)^2}{1-q} q\]
\end{lemma}

Based on these bounds, the following technical lemma provides some connection between the chi-square distance and the KL-divergence of Bernoulli distributions.

\begin{lemma} \label{lemma:chi-square-pseudolinear-with-dkl}
    Let $p, q \in [0,1]$. Then $\chi^2(p,q) \ge \frac{1}{12} \kl(p,q) / \log \max\{\frac{1}{q}, \frac{1}{1-q}\}$.
\end{lemma}

%Finally, we recall and prove Lemma \ref{lemma:chi-square-pseudolinear-with-dkl}:
%\lemmaZchiZsquareZpseudolinearZwithZdkl*
\begin{proof}
    We actually prove that $\chi^2(p,q) / \kl(p,q) \ge \frac{1}{12 \log \max\{\frac{1}{q}, \frac{1}{1-q}\}}$ for every $p\in [0,1]$ and $q \in (0,1)$. For the edge cases of $q$ (which we do not use in our proofs anyway except when $p=q$), the bound remains correct by considering the limit of $\kl(p,q) / \log \max\{\frac{1}{q}, \frac{1}{1-q}\}$.
    
    Without loss of generality, $q \le \frac{1}{2}$. We can assume so since $\chi^2(p,q) = \chi^2(1-p, 1-q)$, $\kl(p,q) = \kl(1 - p, 1 - q)$ and $\max\{\log 1/q, \log 1/(1-q)\} = \max\{\log 1/(1-q), \log 1/(1-(1-q))\}$. Based on this assumption it is sufficient to show that $\frac{\chi^2(p,q)}{\kl(p,q)} \ge \frac{1}{12 \log q^{-1}}$.

    Let $a = p/q$ ($0 < a \le 1/q$). If $a \ge 2$:
    \begin{eqnarray*}
        \chi^2(aq,q) &\ge& \frac{(a-1)^2}{2(a+1)}q \\
        \kl(aq,q) &\le& p \log \frac{p}{q} \le a \log q^{-1} \cdot q \\
        \frac{\chi^2(aq,q)}{\kl(aq,q)} &\ge& \frac{(a-1)^2}{2a(a+1) \log q^{-1}} \ge \frac{1}{12 \log q^{-1}}
    \end{eqnarray*}
    If $a \le 2$:
    \begin{eqnarray*}
        \chi^2(aq,q) &\ge& \frac{(a-1)^2}{2(a+1)}q \\
        \kl(aq,q) &\le& \frac{(a-1)^2}{1-q} q \\
        \frac{\chi^2(aq,q)}{\kl(aq,q)} &\ge& \frac{1-q}{2(a+1)} \ge \frac{1}{12 \log q^{-1}}
    \end{eqnarray*}
    Where the second equation follows from Lemma \ref{lemma:KL-bound-aq-q}, and the very last inequality in the third equation uses the assumption that $q\le\frac12$.
\end{proof}

At this point we prove a version of Lemma \ref{lemma:closeness-linearithmic-fd-soundness} for a parameterized restricted case. We then reduce the general case to this lemma using an appropriate parameter. We note here that if we were to use the following parameterized lemma directly instead of Lemma \ref{lemma:closeness-linearithmic-fd-soundness}, we would have obtained a direct generalization of the binary setting part of \cite[Theorem 4.1]{horrible}.

\begin{lemma} \label{lemma:closeness-linearithmic-fd-soundness-if-nice}
    Let $0 < q < \frac{1}{2}$. Assume that for every $1 \le i \le n$ and for every condition $w \in \{0,1\}^{i-1}$, the prefix marginal $\mu|_i^{x_{[i-1]} = w}$ is equivalent to $\tau|_i^{x_{[i-1]} = w}$ or equivalent to $\Ber(p_{i,w})$ for some $q \le p_{i,w} \le 1 - q$ (or both). Then $\Delta_{\chi^2}(\tau, \mu) \ge \frac{1}{6} (\dtv(\tau,\mu))^2 / \log q^{-1}$.
\end{lemma}

\begin{proof}
    By Pinsker's inequality and Lemma \ref{lemma:chain-rule-dkl-01n} we obtain:
    \[2(\dtv(\tau, \mu))^2 \le \dkl(\tau, \mu)
    = \sum_{i=1}^n \E_{w \sim \tau} \left[\kl\left(\tau|_i^{x_{[i-1]} = w_{[i-1]}}(1), \mu|_i^{x_{[i-1]} = w_{[i-1]}}(1)\right)\right] \]
    By our assumption, for every $1\leq i\leq n$, if $\dkl(\tau|_i^{x_{[i-1]} = w_{[i-1]}}, \mu|_i^{x_{[i-1]} = w_{[i-1]}}) \neq 0$ then the probability of $\mu|_i^{x_{[i-1]} = w_{[i-1]}}$ to draw $1$ is between $q$ and $1 - q$. By Lemma \ref{lemma:chi-square-pseudolinear-with-dkl} we obtain:
    \[2(\dtv(\tau, \mu))^2 \le \sum_{i=1}^n \E_{w \sim \tau} \left[\chi^2\left(\tau|_i^{x_{[i-1]} = w_{[i-1]}}(1), \mu|_i^{x_{[i-1]} = w_{[i-1]}}(1)\right) \cdot 12 \log q^{-1} \right] \]
    That is,
    \[\frac{(\dtv(\tau, \mu))^2}{6\log q^{-1}}
    \le \sum_{i=1}^n \E_{w \sim \tau} \left[\chi^2\left(\tau|_i^{x_{[i-1]} = w_{[i-1]}}(1), \mu|_i^{x_{[i-1]} = w_{[i-1]}}(1)\right)  \right]
    = \Delta_{\chi^2}(\tau, \mu) \qedhere \]
\end{proof}

We could bound the KL-divergence of single bits using the $\chi^2$ distance only because we assumed that the marginal probabilities of $\mu$ are not too close to $0$ or $1$ (unless they are equal to their counterparts in $\tau$). In the general case we cannot assume it, hence we need to instead consider a distribution $\mu'$ which is close to $\mu$ while satisfying this assumption.

We consider $\mu$ as a ``probability tree'', where the root represents the empty string, every edge represents an additional bit, and every leaf represents a complete sample. This tree (and hence the distribution $\mu$) is fully determined using probabilities of the form $\Pr_{x \sim \mu}[x_i = 1 | x_{[i-1]} = w]$, where $1 \le i \le n$ and $w \in \{0,1\}^{i-1}$.

We construct another distribution $\mu'$ based on such a tree pattern. For every $1 \le i \le n$ and $w \in \{0,1\}^{i-1}$, we set $\Pr_{x \sim \mu'}[x_i = 1 | x_{[i-1]} = w]$ as follows:
\begin{eqnarray*}
    \min\left\{\frac{\dtv(\tau, \mu)}{2n}, \Pr_{x \sim \tau}[x_i = 1 | x_{[i-1]} = w]\right\} &\mathrm{if}& \Pr_{x \sim \mu}[x_i = 1 | x_{[i-1]} = w] < \frac{\dtv(\tau,\mu)}{2n} \\
    1 - \min\left\{\frac{\dtv(\tau, \mu)}{2n}, \Pr_{x \sim \tau}[x_i = 0 | x_{[i-1]} = w]\right\} &\mathrm{if}& \Pr_{x \sim \mu}[x_i = 1 | x_{[i-1]} = w] > 1- \frac{\dtv(\tau,\mu)}{2n} \\
    \Pr_{x \sim \mu}[x_i = 1 | x_{[i-1]} = w] \phantom{\bigg\}} && \mathrm{otherwise}
\end{eqnarray*}
Observe that for every $1 \le i \le n$ and $w \in \{0,1\}^{i-1}$, $\dtv\left(\mu|_i^{x_{[i-1]} = w}, \mu'|_i^{x_{[i-1]} = w}\right) \le \frac{\dtv(\tau,\mu)}{2n}$. Hence, $\dtv(\mu,\mu') \le \Delta_\mathrm{TV}(\mu,\mu') \le \frac{1}{2}\dtv(\tau,\mu)$. By the triangle inequality, $\dtv(\tau,\mu') \ge \frac{1}{2}\dtv(\tau,\mu)$.

Since the assumptions of Lemma \ref{lemma:closeness-linearithmic-fd-soundness-if-nice} hold for $\mu'$ (with $q = \frac{1}{2n}\dtv(\tau,\mu)$), we can now conclude the proof of Lemma \ref{lemma:closeness-linearithmic-fd-soundness}:
{
\allowdisplaybreaks
\begin{align*}
    \Delta_{\chi^2}(\tau, \mu)
    &=&& \sum_{i=1}^n \E_{w \sim \tau} \left[\chi^2\left(\tau|_i^{x_{[i-1]} = w_{[i-1]}}(1), \mu|_i^{x_{[i-1]} = w_{[i-1]}}(1)\right)\right] \\
    \text{[Monotonicity of $\chi^2$]}
    &\ge&& \sum_{i=1}^n \E_{w \sim \tau} \left[\chi^2\left(\tau|_i^{x_{[i-1]} = w_{[i-1]}}(1), \mu'|_i^{x_{[i-1]} = w_{[i-1]}}(1)\right)\right] \\
    \text{[Lemma \ref{lemma:closeness-linearithmic-fd-soundness-if-nice} with $\textstyle q = \frac{1}{2n}\dtv(\tau,\mu)$]}
    &\ge&& \lfrac{(\dtv(\tau,\mu'))^2}{6 \log \frac{2n}{\dtv(\tau,\mu)}} \\
    \text{[$\textstyle\dtv(\tau,\mu') \ge \frac{1}{2}\dtv(\tau,\mu)$]}
    &\ge&& \lfrac{(\dtv(\tau,\mu))^2}{24 \log \frac{2n}{\dtv(\tau,\mu)}} & \!\!\!\!\qed
\end{align*}
}

\subsection{Proof of Lemma \ref{lemma:single-bit-chi-square-test}}
\label{sec:closeness-main-path:subsec:chi-square-test}

Here we prove the lemmas deferred from Section \ref{sec:closeness-main-path}. We start with some helper lemmas.

\begin{lemma} \label{lemma:chi-square-helper-1}
    Let $0 \le p \le q \le 1$ for which $p + q \le 1$ and let $X$ be the sum of $N$ independent bits drawn from $\Ber(p)$. Then,
    \[
        \Pr[X \ge \frac{1}{2}(p + q)N] \le e^{-\frac{1}{12} \chi^2(p, q)} 
    \]
\end{lemma}
\begin{proof}
    Let $\delta = q - p$. If $\delta \le 2p$ then by Chernoff's bound:
    \begin{eqnarray*}
        \Pr[X \ge \frac{1}{2}(p + q)N]
        &=& \Pr[X \ge (1 + \frac{\delta}{2p}) \E[X]] \\
        &\le& e^{-\frac{\delta^2}{12p^2} pN}
        = e^{-\frac{\delta^2}{12p} N}
        \le e^{-\frac{(p-q)^2}{12(p + q)(2 - (p+q))} N}
        = e^{-\frac{1}{12}\chi^2(p,q) N}
    \end{eqnarray*}
    If $\delta > 2p$ then:
    \begin{eqnarray*}
        \Pr[X \ge \frac{1}{2}(p + q)N]
        &=& \Pr[X \ge (1 + \frac{\delta}{2p}) \E[X]] \\
        &\le& e^{-\frac{\delta}{6p} pN}
        = e^{-\frac{\delta}{6} N}
%        = e^{-\frac{\dtv(\Ber(p),\Ber(q))}{6} N}
        \le e^{-\frac{(q-p)^2}{6(p+q)}}
        \le e^{-\frac{1}{12} \chi^2(p,q) N}
    \end{eqnarray*}  
    
    \vspace{-\baselineskip-\parskip}   
\end{proof}

\begin{lemma} \label{lemma:chi-square-helper-2}
    Let $0 \le p \le q \le 1$ for which $p + q \le 1$ and let $X$ be the sum of $N$ independent bits drawn from $\Ber(q)$. Then,
    \[
        \Pr[X \le \frac{1}{2}(p + q)N] \le e^{-\frac{1}{8} \chi^2(p, q)} \le e^{-\frac{1}{12} \chi^2(p, q)}
    \]
\end{lemma}
\begin{proof}
    Let $\delta = q - p \le q$. By Chernoff's bound:
    \begin{eqnarray*}
        \Pr[X \le \frac{1}{2}(p + q)N]
        &=& \Pr[X \le (1 - \frac{\delta}{2q}) \E[X]] \\
        &\le& e^{-\frac{\delta^2}{8q^2} qN}
        = e^{-\frac{\delta^2}{8q} N}
        \le e^{-\frac{(p-q)^2}{8(p + q)(2 - (p+q))} N}
        = e^{-\frac{1}{8}\chi^2(p,q) N}
    \end{eqnarray*}
    
    \vspace{-\baselineskip-\parskip}
\end{proof}

\begin{lemma} \label{lemma:chi-square-helper-3}
    Let $p,q \in [0,1]$. Let $X$ be the sum of $N$ independent bits drawn from $\Ber(p)$ and $Y$ be the sum of $N$ independent bits drawn from $\Ber(q)$. Then, with probability at least $\left(1 - e^{-\frac{1}{12}\chi^2(p,q)N}\right)^2$, the sign of $X-Y$ matches the sign of $p-q$.
\end{lemma}
\begin{proof}
    Without loss of generality, $p + q \le 1$ (otherwise use $1-q$ and $1-p$ instead of $p$ and $q$, noting that $d(1-q,1-p) = d(p,q)$ and that the sign of $(1-q)-(1-p)$ matches the sign of $p-q$).

    If $p = q$ then the lemma is vacuously correct. If $p < q$, then by Lemma \ref{lemma:chi-square-helper-1},
    \begin{eqnarray*}
        \Pr[X - Y < 0] = \Pr[X < Y] &\ge& \Pr[X < \frac{1}{2}(p + q)N] \Pr[Y > \frac{1}{2}(p + q)N] \\
        &\ge& \left(1 - e^{\frac{1}{12} \chi^2(p,q) N} \right)^2
    \end{eqnarray*}
    The case where $p > q$ is analogous, using Lemma \ref{lemma:chi-square-helper-2}.
\end{proof}

We recall and prove Lemma \ref{lemma:single-bit-chi-square-test}:

\lemmaZsingleZbitZchiZsquareZtest*

\begin{proof}
    We repeat the following procedure $64$ times: let $N = \ceil{24 / \eps}$. Also, let $X$ be the sum of $N$ independent samples drawn from $\Ber(p)$ and $Y$ be the sum of $N$ independent samples drawn from $\Ber(q)$. To conclude a single trial we check whether $A>B$ or $A<B$ (or neither) holds.

    Let $A$ be the number of trials with $X>Y$ and $B$ be the number of trials with $X<Y$.
    If $|A|,|B| \le 40$ we accept ($p=q$), and otherwise we reject ($\chi^2(p,q) > \eps$).

    If $p=q$, then by symmetry, $\E[X < Y] = \E[X > Y] \le \frac{1}{2}$. That is, $\E[A] = \E[B] \le 32$. By Chernoff's bound, $\Pr[A \ge 41] < e^{-2 \cdot 9^2 / 64} < \frac{1}{6}$ and $\Pr[B \ge 41] < \frac{1}{6}$. Hence, the probability to accept is at least $1 - \frac{2}{6} = \frac{2}{3}$.

    If $\chi^2(p,q) > \eps$, then by Lemma \ref{lemma:chi-square-helper-3}, one of $\Pr[X < Y]$ and $\Pr[X > Y]$ is at least $(1 - e^{-\frac{1}{12} \chi^2(p,q)N})^2 \ge (1 - e^{-2})^2 > 0.74$. Without loss of generality, we assume that $p < q$. In this case, $\E[A] > 47.36$. By the Chernoff bound, the probability to reject is at least $1-\Pr[A \le 40] \ge 1-e^{-2 \cdot 7.36^2/64} > \frac{2}{3}$.
\end{proof}

\section{Lower bound for \propname{product}}
\label{sec:lbnd-product-main-path}

This section is devoted to the following lower bound:

\thZlbndZindependence

Let $\pi_n$ denote the uniform distribution over $\{0,1\}^n$. We show that distinguishing between $\pi_n$ (which is in particular a product distribution) and a distribution that is $\eps$-far from every product distribution requires $\tilde{\Omega}(\sqrt{n}/\eps^2)$ many queries.

Before we present our construction, we cite the corresponding lower bound for uniformity of a product distribution \cite{baynets17}:
\begin{lemma}[\cite{baynets17}] \label{lemma:bayesnets-sqrt-lbnd-uniformity}
    Let $\mathcal N$ be the following distribution over inputs: draw $b_1,\ldots,b_n \sim \{+1, -1\}$ uniformly and independently, and return the distribution $\prod_{i=1}^n \Ber\left(\frac{1}{2} + b_i \frac{\eps}{\sqrt{n}} \right)$. Then the drawn input is always $\Omega(\eps)$-far from $\pi_n$, and any unconditional sampling algorithm that distinguishes between inputs drawn from $\mathcal N$ and $\pi_n$ must take at least $\Omega(\sqrt{n} / \eps^2)$ many samples.
\end{lemma}
In our construction, instead of adding a random bias for each coordinate, we partition the coordinates into pairs, and in each pair introduce a random ``anti-product bias'' as follows.

For $b \in \{0, +1, -1\}$, let $\nu_b$ be the following distribution over $\{0,1\}^2$:
\[\begin{array}{lcrclcr}
    \nu_b(00) &=& \frac{1}{4} + b_i \frac{\eps}{\sqrt{n}} &
    \nu_b(01) &=& \frac{1}{4} - b_i \frac{\eps}{\sqrt{n}} \\
    \nu_b(10) &=& \frac{1}{4} - b_i \frac{\eps}{\sqrt{n}} &
    \nu_b(11) &=& \frac{1}{4} + b_i \frac{\eps}{\sqrt{n}}
\end{array}\]
That is, $\nu_0$ is the uniform distribution over two bits, and $\nu_{+1}$, $\nu_{-1}$ are non-product distributions over two bits.

Let $\mathcal Y$ be the distribution over inputs that always returns the uniform distribution over $\{0,1\}^n$, and let $\mathcal{N}$ be the following distribution over ``bad'' inputs over $\{0,1\}^n$: we partition $[n]$ into fixed pairs (for example, $(2i-1,2i)$ for every $1 \le i \le \floor{n/2}$). For every pair we draw $b_i \in \{+1, -1\}$ uniformly and independently. The distribution of the $i$th pair is $\nu_{b_i}$. All pairs are independent of each other. We assume that $n$ is even, since for an odd $n$ we can just assume that the $n$th bit is distributed uniformly and independently of the rest.

The following lemma, whose proof appears in Subsection \ref{sec:lbnd-product-main-path:subsec:technical-proofs-product}, states that $\mathcal N$ always draws a distribution far from \propname{product}.

\begin{restatable}{lemma}{lemmaZsqrtZNZfarZfromZproduct} \label{lemma:sqrt-N-far-from-product}
    Every input distribution drawn from $\mathcal N$ is $\Omega(\eps)$-far from any product distribution.
\end{restatable}

We now prove the indistinguishability of $\mathcal N$ from the uniform distribution, starting with the sampling model, to which we later show a reduction from the subcube conditional model.

\begin{lemma} \label{lemma:lbnd-sqrt-N-hard}
    An unconditional sampling algorithm that distinguishes between the uniform distribution and inputs that are drawn from $\mathcal N$ must make $\Omega(\sqrt{n} / \eps^2)$ many queries.
\end{lemma}
\begin{proof}
    Without loss of generality, assume that $n$ is even. Consider the following bijection over $\{0,1\}^2$: $(x,y) \to (x \oplus y, y)$. For $b \in \{+1, 0, -1\}$, Let $\nu'_b$ be the distribution that draws $(x,y) \sim \nu_b$ and returns $(x \oplus y, y)$. Observe that $\nu'_b$ is identical to the product distribution $\Ber(\frac{1}{2} + \frac{2b\eps}{\sqrt{n}}) \times \Ber(\frac{1}{2})$. Let $f : \{0,1\}^n \to \{0,1\}^n$ be a bijection that applies the above mapping for every pair individually.

    Let $\mathcal N'$ be the distribution over inputs that draws $\mu \sim \mathcal N$, and then returns the distribution that draws $x \sim \mu$ and returns $f(x)$. Note also that the distribution that drawn $x \sim \pi_n$ and returns $f(x)$ is identical to $\pi_n$. Since $f$ is a bijection, this means that the number of samples needed to distinguish between the uniform distribution and $\mathcal N$ is exactly the same as the number of samples needed to distinguish between the uniform distribution and $\mathcal N'$.

    Observe that the form of a distribution resulting from $\mathcal N'$ is $\prod_{i=1}^{n/2} \left(\Ber\left(\frac{1}{2} + \frac{2b_i\eps}{\sqrt{n}}\right) \times \Ber\left(\frac{1}{2}\right)\right)$. The restriction of $\mathcal N'$ to odd indexes is identical to the construction of Lemma \ref{lemma:bayesnets-sqrt-lbnd-uniformity} (with parameters $\sqrt{2}\eps$ and $n/2$), which requires at least $\Omega(\sqrt{n} / \eps^2)$ many unconditional samples to distinguish between it and the uniform distribution. Hence, this is also a lower bound for the number of samples needed to distinguish between the uniform distribution and inputs drawn from $\mathcal N$.
\end{proof}

\begin{lemma} \label{lemma:product-lbnd-sqrt-simulation-of-a-pair}
    Let $b \in \{0,+1,-1\}$. Then for every subcube restriction $q \in \{0,1,*\}^2$, we can simulate a $q$-conditioned sample from $\nu_b$ by drawing a single unconditional sample from $\nu_b$, without having any knowledge about $b$.
\end{lemma}
\begin{proof}
    Let $q \in \{0,+1,*\}^2$ and let $b \in \{0,+1,-1\}$. We first draw an unconditional sample $x \sim \nu_b$.
    There are three kinds of subcube restrictions:
    \begin{itemize}
        \item Trivial: $q \in \{00, 01, 10, 11\}$. To simulate such a query, we simply return $q$, ignoring the unconditional sample we have.
        \item Unconditional: $q = **$. To simulate such a query, we simply return $x$.
        \item Single restriction: $q \in \{0*, 1*, *0, *1\}$. We denote the parity of $x$ by $p = x|_1 \oplus x|_2$, and construct the output as the concatenation of the two bits $\mathrm{out}_1$ and $\mathrm{out}_2$, based on the following table:
        \[\begin{array}{rllll}
             q & 0* & 1* & *0 & *1 \\
             \mathrm{out}_1 & 0 & 1 & p & 1\oplus p \\ 
             \mathrm{out}_2  & p & 1\oplus p & 0 & 1 \\
        \end{array}\]
    \end{itemize}
    The correctness of the output in the trivial case and the unconditional case is trivial. We prove the correctness of the single-restriction case only for $q = 0*$, since the other cases are analogous.
    \[\Pr_{\nu_b}[00 | 0*]
    = \frac{\Pr_{\nu_b}[00]}{\Pr_{\nu_b}[0*]}
    = \frac{\frac{1}{4} + \frac{b_i \eps}{\sqrt{n}}}{\frac{1}{4} + \frac{b_i \eps}{\sqrt{n}} + \frac{1}{4} - \frac{b_i \eps}{\sqrt{n}}}
    = \frac{1}{2} + 2 \frac{b_i \eps}{\sqrt{n}}
    = \Pr_{\nu_b}[00 \vee 11]
    = \Pr_{x\sim \nu_b}[x|_1 \oplus x|_2 = 0] \qedhere\]
\end{proof}
At this point we can prove Theorem \ref{th:lbnd-independence}.
\begin{proof}[Proof of Theorem \ref{th:lbnd-independence}]
    By Lemma \ref{lemma:sqrt-N-far-from-product}, $\mathcal N$ draws an input distribution that is $\Omega(\eps)$-far from any product distribution.
    Observe that since the structure of the pairs is known in advance and since they are independent, we can simulate every $q$-subcube-query algorithm using a $q$-unconditional-sample algorithm: for each query we draw a single sample and then use the simulation procedure of Lemma \ref{lemma:product-lbnd-sqrt-simulation-of-a-pair} for every pair in itself.
    
    Since an unconditional test requires $\Omega(\sqrt{n} / \eps^2)$ queries to distinguish inputs drawn from $\mathcal{N}$ from the uniform distribution by Lemma \ref{lemma:lbnd-sqrt-N-hard}, and every subcube query to the uniform distribution or an input drawn from $\mathcal{N}$ can be simulated using a single unconditional query, the lower bound holds for subcube algorithms as well.
\end{proof}

\subsection{Technical proofs for the \propname{product} lower bound}
\label{sec:lbnd-product-main-path:subsec:technical-proofs-product}

Recall that we denote the uniform distribution over $\{0,1\}^n$ by $\pi_n$.

Before we prove Lemma \ref{lemma:sqrt-N-far-from-product}, we need the following technical lemmas.

\begin{lemma} \label{lemma:binom-bound}
    For a string $y \in \{0,1\}^n$, let $C_1(y)$ be the number of $1$s in $y$. If $n \ge 70$, then $\Pr_{y\sim \pi_n}[C_1(y) > \frac{1}{2}n + \frac{1}{4}\sqrt{n}] \ge \frac{1}{4}$.
\end{lemma}
\begin{proof}
    We use the well-known bound $\binom{n}{\floor{n/2}} \le 2^n \cdot \sqrt{\frac{2}{\pi n}}$ (for all $n \ge 1$) to obtain for $n\ge 70$:
    
    \begin{eqnarray*}
        \Pr_{y\sim \pi_n}\left[\abs{C_1(y) - \frac{1}{2}n} \le \frac{1}{4}\sqrt{n}\right]
        &=& 
        \sum_{k = \ceil{\frac{1}{2}n - \frac{1}{4} \sqrt{n}}}^{\floor{\frac{1}{2}n + \frac{1}{4} \sqrt{n}}} \Pr_{y\sim \pi_n}\left[C_1(y) = k\right] \\
        \text{[Since $\textstyle{\binom{n}{k}\le \binom{n}{\floor{n/2}}}$]}
        &\le& \left(2 \cdot \frac{1}{4}\sqrt{n} + 1\right) 2^{-n} \binom{n}{\floor{n/2}} \\
        &\le& \left(\frac{1}{2}\sqrt{n} + 1\right) \frac{\sqrt{2}}{\sqrt{\pi n}}
        \le \frac{1}{2}
    \end{eqnarray*}

    By symmetry reasons, if $n\ge 70$ then:
    \begin{eqnarray*}
        \Pr_{y\sim \pi_n}\left[C_1(y) > \frac{1}{2}n + \frac{1}{4}\sqrt{n}\right]
        &=& \frac{1}{2}\Pr_{y\sim \pi_n}\left[\abs{C_1(y) - \frac{1}{2}n} > \frac{1}{4}\sqrt{n}\right] \\
        &=& \frac{1}{2}\left(1 - \Pr_{y\sim \pi_n}\left[\abs{C_1(y) - \frac{1}{2}n} \le \frac{1}{4}\sqrt{n}\right]\right)
        \ge \frac{1}{4}
    \end{eqnarray*}
\end{proof}

\begin{lemma} \label{lemma:marginal-ratio-to-dtv}
    Let $0 < \delta < \frac{1}{4 \sqrt{n}}$ and $\tau$ be a product distribution over $\{0,1\}^n$. For sufficiently large $n$, if $\abs{\Pr_\tau[x_i = 1] - \frac{1}{2}} > \delta$ for every $1 \le i \le n$, then the distance of $\tau$ from the uniform distribution over $\{0,1\}^n$ is at least $\frac{1}{16} \delta \sqrt{n}$.
\end{lemma}

\begin{proof}
    Without loss of generality, we can assume that $\Pr_\tau[x_i = 1] \ge \frac{1}{2}$ for every $1 \le i \le n$. Otherwise, we can negate the ``wrong'' bits while preserving the distance from the uniform distribution. Based on this assumption, we have a product distribution whose probability to draw $1$ at any individual index is at least $\frac{1 + \delta}{2}$.

    Let $\tau'$ be the product distribution whose probability to draw $1$ at any index is exactly $\frac{1+\delta}{2}$. Observe that the distance of $\tau'$ from the uniform distribution is a lower bound of the distance of $\tau$ from the uniform distribution.

    Let $y \in \{0,1\}^n$ be a string, and let $C_1(y)$ (respectively $C_0(y)$) be the number of $1$s (respectively $0$s) in $y$. If $C_1(y) \ge \frac{1}{2}n + \frac{1}{4}\sqrt{n}$, then:
    \begin{eqnarray*}
        \frac{\Pr_{\tau'}[y]}{\Pr_{\pi_n}[y]}
        &=& (1 + \delta)^{C_1(y)} \cdot (1 - \delta)^{n - C_1(y)} \\
        &=& (1 + \delta)^{C_1(y) - C_0(y)} \cdot ((1 + \delta)(1 - \delta))^{C_0(y)} \\
        &\ge& (1 + \delta)^{\sqrt{n} / 2} \cdot ((1 + \delta)(1 - \delta))^{\frac{1}{2}n - \frac{1}{4}\sqrt{n}} \\
        &=& (1 + \delta)^{\sqrt{n} / 2} \cdot (1 - \delta^2)^{\frac{1}{2}n - \frac{1}{4}\sqrt{n}} \\
        &\ge& (1 + \delta)^{\sqrt{n} / 2} \cdot (1 - \delta^2)^{\frac{1}{2}n} \\
        \text{[$(1+a)^b \ge 1+ab$ for $b \ge 1$, $|a|<1$]}&\ge& \left(1 + \frac{1}{2}\delta\sqrt{n}\right) \cdot  \left(1 - \frac{1}{2}\delta^2 n\right) \\
        \text{[$\textstyle \delta < \frac{1}{4\sqrt{n}}$]} &\ge& 1 + \frac{1}{4} \delta\sqrt{n}
    \end{eqnarray*}

    In particular,
    \[ \frac{\Pr_{y \sim \tau'}\left[C_1(y) > \frac{1}{2}n + \frac{1}{4}\sqrt{n}\right]}{\Pr_{y \sim \pi_n}\left[C_1(y) > \frac{1}{2}n + \frac{1}{4}\sqrt{n}\right]} \ge 1+\frac{1}{4}\delta\sqrt{n} \]

    Since $\Pr_{y \sim \pi_n}\left[C_1(y) > \frac{1}{2}n + \frac{1}{4}\sqrt{n}\right] \ge \frac{1}{4}$ by Lemma \ref{lemma:binom-bound}, we obtain that $\dtv(\tau', \pi_n) \ge \frac{1}{4}\delta\sqrt{n} \cdot \frac{1}{4} = \frac{1}{16}\delta\sqrt{n}$.
\end{proof}

We recall and prove Lemma \ref{lemma:sqrt-N-far-from-product}:

\lemmaZsqrtZNZfarZfromZproduct*

\begin{proof}
    Without loss of generality, we assume that $n$ is even. Let $b_1,\ldots,b_{n/2} \in \{+1, -1\}$ be the parameters of the construction, that is, the drawn input is $\mu = \prod_{i=1}^{n/2} \nu_{b_i}$. Recall that for $b \in \{+1, -1\}$, $\Pr_{\nu_b}[x_1 = 1] = \Pr_{\nu_b}[x_2 = 1] = \frac{1}{2}$ and $\Pr_{\nu_b}[x_1 \oplus x_2 = 1] = \frac{1}{2} - \frac{2b \eps}{\sqrt{n}}$. We assume that $n$ is sufficiently large so that $\floor{n/6}$ satisfies the constraints of Lemma \ref{lemma:marginal-ratio-to-dtv}.

    For a given product distribution $\tau$, for every $1 \le i \le n/2$, let:
    \begin{eqnarray*}
        \delta_{i,0} &=& \Pr_\tau[x_{2i} = 1] - \frac{1}{2} \\
        \delta_{i,1} &=& \Pr_\tau[x_{2i-1} = 1] - \frac{1}{2} \\
        \delta_{i,2} &=& \Pr_\tau[x_{2i} \oplus x_{2i-1} = 1] - \frac{1}{2}
    \end{eqnarray*}

    Let $I_0$ be the set of indexes for which $\abs{\delta_{i,0}} > \frac{\eps}{10\sqrt{n}}$. Let $I_1$ be defined analogously for $\delta_{i,1}$, and let $I_2 = [n/2] \setminus (I_0 \cup I_1)$ be the set of all other indexes. Since $|I_0| + |I_1| + |I_2| = n/2$, at least one of them contains at least $\frac{1}{6}n$ elements.

    \paragraph{Case I. $|I_0| \ge \frac{1}{6}n$} Let $I = \{ 2i : i \in I_0\}$. Observe that $\mu|_I$ is uniform over $\{0,1\}^{n/2}$, since $\mu$ is the product of $n/2$ independent distributions over pairs whose marginals are $\frac{1}{2}$, and every pair contributes exactly one index. According to Lemma \ref{lemma:marginal-ratio-to-dtv},
    \[\dtv(\tau,\mu)
    \ge \dtv(\tau|_I, \mu|_I)
    \ge \frac{1}{16} \cdot \frac{\eps}{10 \sqrt{n}} \cdot \sqrt{\frac{1}{6}n}
    = \frac{1}{160\sqrt{6}}\eps > \frac{1}{400}\eps\]

    \paragraph{Case II. $|I_1| \ge \frac{1}{6}n$} Completely analogous to Case I. Again $\dtv(\tau,\mu)>\frac{1}{400}\eps$.
    
    \paragraph{Case III. $|I_2| \ge \frac{1}{6}n$} For every $i \in I_2$, $|\delta_{i,0}|, |\delta_{i,1}| \le \frac{\eps}{10\sqrt{n}}$. Hence,
    \begin{eqnarray*}
        |\delta_{i,2}|
        = \abs{\Pr_\tau[x_{2i-1} \oplus x_{2i} = 1] - \frac{1}{2}}
        &=& \abs{\left(\frac{1}{2} - \delta_{i,0}\right)\left(\frac{1}{2} + \delta_{i,1}\right) + \left(\frac{1}{2} + \delta_{i,0}\right)\left(\frac{1}{2} - \delta_{i,1}\right) - \frac{1}{2}} \\
        &=& 2|\delta_{i,0}| |\delta_{i,1}|
        \le \frac{\eps^2}{50n}
    \end{eqnarray*}
    
    Let $I = \{ 2i : i \in I_2 \} \cup \{ 2i - 1 : i \in I_2 \}$
    Let $f : \{0,1\}^I \to \{0,1\}^{|I_2|}$ be the function that maps every pair in $I_2$ to its parity bit. In other words, for every $i \in I_2$, the bits $x_{2i-1}, x_{2i}$ are mapped to a single bit $x_{2i-1} \oplus x_{2i}$. Let $\mu'$ (respectively $\tau'$) be the distribution over $\{0,1\}^{|I_2|}$ that draws a sample $x \sim \mu$ (respectively $x \sim \tau$) and returns $f(x)$.
    Observe that both $\mu'$ and $\tau'$ are product distributions over $\{0,1\}^{|I_2|}$, since the pairs are independent and every pair is mapped into a single bit.

    Note that $\dtv(\tau', \pi_{|I_2|}) \le \Delta_\mathrm{TV}(\tau', \pi_{|I_2|}) = \sum_{i=1}^{|I_2|} \dtv\left(\tau'|_i, \Ber(1/2)\right) \le |I_2| \cdot \frac{\eps^2}{50n} \le \frac{1}{50}\eps^2$.

    In $\mu'$, by definition of the $\nu_b$s, all marginals have the form $\frac{1}{2} \pm \frac{\eps}{\sqrt{n}}$, hence by Lemma \ref{lemma:marginal-ratio-to-dtv}, 
    \[\dtv(\mu', \pi_{|I_2|})
    \ge \frac{1}{16} \cdot \frac{\eps}{\sqrt{n}} \cdot \sqrt{\frac{1}{6}n} 
    = \frac{\eps}{16\sqrt{6}}
    \ge \frac{1}{40}\eps\]

    By a data processing inequality and the triangle inequality,
    \[\dtv(\mu, \tau) \ge \dtv(\mu', \tau') \ge \dtv(\mu', \pi_{|I_2|}) - \dtv(\tau', \pi_{|I_2|}) \ge \frac{1}{40}\eps - \frac{1}{50}\eps^2 > \frac{1}{400}\eps\qedhere\]
\end{proof}

% \section{Technical proofs for \propname{equivalence} testing}
% \label{sec:technical-proofs-closeness}

\section{Extending the \propname{equivalence} test to general alphabets}
\label{sec:closeness-omegas}

We extend the definitions of the prefix oracle to non-binary settings.

\begin{definition}[Subcube oracle access in non-binary strings]
    Let $\mu$ be an unknown distribution over $\prod_{i=1}^n \Omega_n$, where $\Omega_1,\ldots,\Omega_n$ are all finite. The \emph{subcube oracle} has as input a tuple $(A_1,\ldots,A_n)$ where $A_i \subseteq \Omega_i$ for every $1 \le i \le n$. The output distributes as $\mu|^{\prod_{i=1}^n A_i}$. For technical reasons, if $\Pr_\mu[\prod_{i=1}^n A_i] = 0$, then the oracle indicates an error. Note that the answers of the oracle are fully independent of the answers that were given to previous queries.
\end{definition}

In the corresponding definition for a prefix oracle, we still demand that until the ``break-off index'' $i$ all conditions force single outcomes from the sets $\Omega_j$, while after the break-off index there are no restrictions at all. However, at index $i$ we allow conditions to any subset of $\Omega_i$ to take place. There is no such distinction in the binary case, where $|\Omega_i|=2$ and hence all non-trivial conditions are to a single outcome.

\begin{definition}[Prefix oracle access in non-binary settings]
    Let $\mu$ be an unknown distribution over $\prod_{i=1}^n \Omega_n$, where $\Omega_1,\ldots,\Omega_n$ are all finite. The input of the \emph{prefix oracle} consists of an index $1 \le i \le n$, which we refer to as the \emph{index of the prefix}, elements $a_j \in \Omega_j$ for every $1 \le j \le i-1$, and a condition $A \subseteq \Omega_i$. The output of the oracle distributes like $\mu|^{\{x : x_i \in A \wedge x_1=a_1,\ldots, x_{i-1} = a_{i-1}\}}$.
\end{definition}

\begin{definition}[Marginal subcube oracle access in non-binary settings]
        Let $\mu$ be an unknown distribution over $\prod_{i=1}^n \Omega_n$, where $\Omega_1,\ldots,\Omega_n$ are all finite. The \emph{marginal subcube oracle} has as input an index $i$ and a tuple $(A_1,\ldots,A_n)$ where $A_j \subseteq \Omega_j$ for every $1 \le j \le n$. The output distributes as $\mu|_i^{\prod_{j=1}^n A_j}$. For technical reasons, if $\Pr_\mu[\prod_{i=1}^n A_i] = 0$, then the oracle indicates an error.
\end{definition}

\begin{definition}[Marginal prefix oracle access in non-binary settings]
    Let $\mu$ be an unknown distribution over $\prod_{i=1}^n \Omega_n$, where $\Omega_1,\ldots,\Omega_n$ are all finite. The input of the \emph{marginal prefix oracle} consists of an index $1 \le i \le n$, which we refer to as the \emph{index of the prefix}, elements $a_j \in \Omega_j$ for every $1 \le j \le i-1$, and a condition $A \subseteq \Omega_i$. The output of the oracle distributes like $\mu|_i^{\{x : x_i \in A \wedge x_1=a_1,\ldots,x_{i-1}=a_{i-1}\}}$.
\end{definition}

Based on these definitions, we state Theorem \ref{th:closeness-omegas}:
\thZclosenessZomegas

Before we prove Theorem \ref{th:closeness-omegas}, we need the following.

\begin{lemma}[Binary form of a composite distribution] \label{lemma:simulation-of-binary-form}
    Let $\mu$ be a distribution over $\prod_{i=1}^n \Omega_i$, where $\Omega_1,\ldots,\Omega_n$ are non-empty finite sets. There exists a distribution $\mu^*$ over $\{0,1\}^{\sum_{i=1}^n \ceil{\log_2 |\Omega_i|}}$ that is equivalent to $\mu$ up to relabeling, for which every subcube (respectively prefix) query to $\mu^*$ can be simulated using a single subcube (respectively prefix) query to $\mu$, and every marginal subcube (respectively prefix) query to $\mu^*$ can be simulated using a single marginal subcube (respectively prefix) query to $\mu$.
\end{lemma}
\begin{proof}
    For every $1 \le i \le n$, let $f_i : \Omega_i \to \{0,1\}^{\ceil{\log_2 |\Omega_i|}}$ be an arbitrary injective function from $\Omega_i$ to $\{0,1\}^{\ceil{\log_2 |\Omega_i|}}$. Let $f : \prod_{i=1}^n \Omega_i \to \{0,1\}^{\sum_{i=1}^n \ceil{\log_2 |\Omega_i|}}$ be the concatenation of these mappings. More precisely, $f((x_1,\ldots,x_n)) = \mathrm{concatenate}(f_1(x_1),\ldots,f_n(x_n))$. Let $N = \sum_{i=1}^n \ceil{\log_2 |\Omega_i|}$.
    
    For every $1 \le i \le n$, we define a projection function $g_i : \{0,1\}^{\sum_{j=1}^n \ceil{\log_2 |\Omega_i|}} \to \Omega_i$ by
    \[g_i(x) = (f_i)^{-1}\left( x|_{\{\sum_{j=1}^{i-1} \ceil{\log_2 |\Omega_j|} + 1,\ldots,\sum_{j=1}^{i} \ceil{\log_2 |\Omega_j|} \}}\right)\]
    where $g_i(x)$ is defined arbitrarily if the corresponding binary string is not in $f_i$'s image, as this will be a zero-probability event. Observe that for every $(x_1,\ldots,x_n) \in \prod_{i=1}^n \Omega_i$ and for every $1 \le i \le n$, $x_i = g_i(f(x_1,\ldots,x_n))$.

    Let $\mu^*$ be the distribution over $\{0,1\}^{\sum_{i=1}^n \ceil{\log_2 |\Omega_i|}}$ that draws $x \sim \mu$ and returns $f(x)$. Since $f$ is an injective function, $\mu^*$ is equivalent to $\mu$ up to relabeling.

    For simulating subcube queries consider some $I^* \subseteq [N]$ and $w^* \in \{0,1\}^{I^*}$. For every $1 \le i \le n$ and $x^* \in \{0,1\}^N$, let $x_i = g_i(x^*)$ (that is, we decompose $x^*$ using $x^* = f(x_1,\ldots,x_n)$). Also, for every $1 \le i \le n$, let $A_i(I^*,w^*) = \left\{g_i(s^*) : s^* \in \{0,1\}^N \wedge s^*|_{I^*} = w^* \right\}$. Based on this composition, we obtain:
    \[ \left\{ x^* \in \{0,1\}^N : x^*|_{I^*} = w^* \right\} = \left\{ f(x_1,\ldots,x_n) : \bigwedge_{i=1}^n \left(x_i \in A_i(I^*, w^*) \right) \right\} \]
    The last expression is the $f$-image of all elements in the $\mu$-subcube condition $\prod_{i=1}^n A_i(I^*,w^*)$. Hence, every subcube query of $\mu^*$ can be simulated using a single subcube query to $\mu$.

    Observe that this construction preserves prefix queries. That is, if a subcube query to $\mu^*$ is a prefix query, then the simulated subcube query to $\mu$ is a prefix query as well. Note that this argument holds for marginal queries as well, since we can extract the relevant bit of the sampled coordinate.
\end{proof}

Theorem \ref{th:closeness-omegas} now follows.

\begin{proof}[Proof of Theorem \ref{th:closeness-omegas}]
    Let $\mu$ and $\tau$ be two distributions over $\prod_{i=1}^n \Omega_i$, where $\Omega_1,\ldots,\Omega_n$ are all finite. We use Lemma \ref{lemma:simulation-of-binary-form} to define two distributions $\mu^*$, $\tau^*$ that are identical to $\mu$, $\tau$ respectively up to relabeling (which is the same in both constructions). Based on this lemma:
    \begin{itemize}
        \item $\dtv(\tau^*, \mu^*) = \dtv(\tau, \mu)$, since $\mu^*$, $\tau^*$ are the same as $\mu$, $\tau$ up to relabeling (which is the same for both constructions $\mu \to \mu^*$ and $\tau \to \tau^*$).
        \item Every prefix query to $\tau^*$ can be simulated using a single prefix query to $\tau$.
        \item Every marginal prefix query to $\mu^*$ can be simulated using a single marginal prefix query to $\mu$.
    \end{itemize}
    Hence, we can distinguish between $\tau = \mu$ and $\dtv(\tau, \mu) > \eps$ using Theorem $\ref{th:closeness-linearithmic}$ with the input $(\mu^*, \tau^*, \eps)$ by simulating every query to $\tau^*$ or $\mu^*$ through a single query to the corresponding input distribution $\tau$ or $\mu$.
\end{proof}

\section{Upper bound for \propname{product}}
\label{sec:ubnd-product}

We reduce a test for \propname{product} to a test for \propname{equivalence}, based on the following key observation:

\begin{observation}
    Let $\mu$ be a distribution over $\prod_{i=1}^n \Omega_i$. If $\mu$ is $\eps$-far from any product distribution, then in particular it is $\eps$-far from the product of its marginals, $\prod_{i=1}^n \mu|_i$.
\end{observation}

In the binary setting, simulating a marginal prefix query to the product of marginals is pretty straightforward and can be done using one unconditional query to $\mu$, which is in particular a prefix query:
\begin{observation} \label{obs:simulation-of-marginal-of-binary-product}
    Let $\mu$ be a distribution over $\{0,1\}^n$, and let $\mu' = \prod_{i=1}^n \mu|_i$ be the product of $\mu$'s marginals. Then we can simulate every marginal prefix query to $\mu'$ using one unconditional sample from $\mu$ (and keeping its $i$th entry).
\end{observation}

In the general setting, we need the stronger subcube access. The reason is that when taking a prefix marginal query at index $i$ from the binary representation of $\mu' = \prod_{j=1}^n \mu|_j$, it may be that that $i$ is ``in the middle'' of the $\lceil\log |\Omega_j|\rceil$-bit representation of $j$th coordinate of $\mu'$, and then this query must translate to a non-trivial set $A_j\subseteq\Omega_j$ when simulating it using query access to $\mu$.

\begin{observation} \label{obs:simulation-of-marginal-of-product}
    Let $\mu$ be a distribution over $\prod_{i=1}^n \Omega_i$, and let $\mu' = \prod_{i=1}^n \mu|_i$ be the product of $\mu$'s marginals. Then we can simulate every marginal prefix query to $\mu'$ using one subcube query to $\mu$.
\end{observation}

At this point, we can prove Theorem \ref{th:ubnd-product}.
\thZubndZproduct
\begin{proof}
    Let $\mu' = \prod_{i=1}^n \mu|_i$ be the product of $\mu$'s marginals. If $\mu$ is a product distribution then $\mu' = \mu$, and if $\mu$ is $\eps$-far from every product distribution, then in particular $\dtv(\mu, \mu') > \eps$.

    Since $\mu$ is accessible through the subcube oracle, we can use Observation \ref{obs:simulation-of-marginal-of-product} to simulate every marginal prefix query to $\mu'$ at the cost of one subcube query to $\mu$. If $|\Omega_i| = 2$ for every $1 \le i \le n$, then we can use an unconditional sample instead of a subcube query by Observation \ref{obs:simulation-of-marginal-of-binary-product}.

    Hence, we can reduce the $\eps$-test of $\mu$ for \propname{product} to an $\eps$-test of the equivalence of $\mu'$ and $\mu$, which we perform using Theorem \ref{th:closeness-omegas}. As noted above, this produces subcube queries for general alphabets, and only prefix queries for the binary setting.
\end{proof}

\bibliographystyle{alpha}
\bibliography{main}

\appendix
%\newpage

\section{Explicit algorithms}
We provide here explicit representations for the binary setting algorithms constructed in this paper.

\paragraph{Levin's work-balance method}
For a random variable $X$ that is bounded between $0$ and $1$, it distinguishes between $\E[X] = 0$ and $\E[X] > \eps$. To do that, we use another random variable $Y$, for which we can distinguish between $\E[X | Y=y] = 0$ and $\E[X | Y=y] > \rho$ for every $0 < \rho < 1$, at a cost of $O(1 / \rho)$. Overall, the cost of the algorithm is $O(\log^2 \eps^{-1} / \eps)$. The original construction is found in \cite{levin1985one} and the following optimized version appears as an exercise in \cite{goldreich2017introduction}.

The black box is run a logarithmic number of times for every $y$ that we draw (as opposed to once in, for example, \cite[Exercise 8.4]{goldreich2017introduction}) since it refers here to a procedure with two-sided error.

\begin{algorithm}[H]
    \caption{Levin's work-balance producedure}
    \label{alg:levin}
    \begin{algorithmic}
        \State \textbf{input} $Y$ -- a random variable, accessible through unconditional sampling.
        \State \textbf{input} $\eps$ -- a threshold parameter.
        \State \textbf{input} A random black box that, for every $y \in \supp(Y)$ and $0 < \eps' < 1$:
        \State \phantom{\textbf{input}} \textbf{completeness} if $\E[X|Y=y] = 0$, it accepts with at least probability $2/3$.
        \State \phantom{\textbf{input}} \textbf{soundness} if $\E[X|Y=y] > \eps'$, it rejects with probability at least $2/3$.
        \State \phantom{\textbf{input}} \textbf{resource cost complexity} $O(1/\eps')$.
        \State \textbf{completeness} If $\E[X] = 0$, then the output is \accept with probability at least $2/3$.
        \State \textbf{soundness} If $\E[X] > \eps$, then the output is \reject with probability at least $2/3$.
        \State \textbf{resource cost complexity} $O(\log^2 \eps^{-1} / \eps)$.
        \For{$t$ from $1$ to $\ceil{\log_2 (2/\eps)}$}
            \State \textbf{let} $\eps' \gets 2^{-t}$.
            \For{$\ceil{2^{3-t} \eps^{-1}}$ \textbf{times}}
                \State \textbf{draw} $y \sim Y$.
                \State \textbf{set} $r \gets 0$.
                \For{$\ceil{64 (\log_2 \eps^{-1} + 2)}$ \textbf{times}}
                    \State \textbf{run} the black box with $(y, \eps')$.
                    \If{the black box accepts}
                        \State \textbf{set} $r \gets r + 1$.
                    \Else
                        \State \textbf{set} $r \gets r - 1$.
                    \EndIf
                \EndFor
                \If{$r < 0$}
                    \State \Return \reject.
                \EndIf
            \EndFor
        \EndFor
        \State \Return \accept.
    \end{algorithmic}
\end{algorithm}

\paragraph{$\chi^2$-test of single-bit distributions}
The following is the algorithm that detects the difference between two Bernoulli distributions as it was described in Lemma \ref{lemma:single-bit-chi-square-test}.

\begin{algorithm}[H]
    \caption{$\chi^2$-test of single-distributions}
    \label{alg:chi-square-single}
    \begin{algorithmic}
        \State \textbf{input} Two Bernoulli distributions $\Ber(p)$ and $\Ber(q)$, accessible through samples.
        \State \textbf{completeness} If $p = q$, then the algorithm accepts with probability at least $2/3$.
        \State \textbf{soundness} If $\chi^2(p, q)>\eps$, then the algorithm rejects with probability at least $2/3$.
        \State \textbf{complexity} $O(1/\eps)$ samples.
        \State \textbf{let} $N \gets \ceil{16 / \eps}$.
        \State \textbf{set} $A,B \gets 0$.
        \For{$64$ \textbf{times}}
            \State \textbf{draw} $N$ independent samples from $\Ber(p)$, let $X$ be their sum.
            \State \textbf{draw} $N$ independent samples from $\Ber(q)$, let $Y$ be their sum.
            \If{$X > Y$}
                \State \textbf{set} $A \gets A + 1$.
            \EndIf
            \If{$X < Y$}
                \State \textbf{set} $B \gets B + 1$.
            \EndIf
        \EndFor
        \If{$A \le 40$ and $B \le 40$}
            \State \Return \accept.
        \Else
            \State \Return \reject.
        \EndIf
    \end{algorithmic}
\end{algorithm}

\paragraph{Testing \propname{equivalence}}
The proof of Theorem \ref{th:closeness-linearithmic} translates to the following explicit algorithm.

\begin{algorithm}[H]
    \caption{$\eps$-test for binary-alphabet \propname{equivalence}}
    \label{alg:equivalence-linearithmic}
    \begin{algorithmic}
        \State \textbf{input} $n$, $\eps > 0$, two distributions $\mu$, $\tau$ over $\{0,1\}^n$.
        \State \phantom{\textbf{input}} $\mu$ is accessible through the marginal prefix oracle.
        \State \phantom{\textbf{input}} $\tau$ is accessible through the prefix oracle.
        \State \textbf{completeness} If $\tau = \mu$, then the output is \accept with probability at least $\frac{2}{3}$.
        \State \textbf{soundness} If $\dtv(\tau, \mu) > \eps$, then the output is \reject with probability at least $\frac{2}{3}$.
        \State \textbf{let} $\pi$ be the uniform distribution over $[n]$.
        \State \textbf{let} $Y$ be a random variable that distributes as $\pi\times\tau$.
        \State \textbf{let} $X$ be a random variable defined as a function of $Y=(i,w)$:
        \[X(i,w)=\chi^2\left(\mu|_i^{x_{[i-1]} = w_{[i-1]}}(1), \tau|_i^{x_{[i-1]} = w_{[i-1]}}(1)\right)\]
        \State \textbf{let} $\rho \gets \frac{\eps^2}{24 \log(\eps/2n)}$.
        \State \textbf{run} Levin's procedure (Algorithm \ref{alg:levin}), where its input consists of $Y$, $\rho/n$, and the black box $((i,w),\eps') \to (\text{Algorithm \ref{alg:chi-square-single} with input $\mu|_i^{x_{[i-1]} = w_{[i-1]}}$, $\tau|_i^{x_{[i-1]} = w_{[i-1]}}$ and $\eps'$})$.
        \If{Levin's procedure accepts}
            \State \Return \accept.
        \Else
            \State \Return \reject.
        \EndIf
    \end{algorithmic}
\end{algorithm}

\paragraph{Testing \propname{product}}
The proof of Theorem \ref{th:ubnd-product} translates in the binary setting to the following explicit algorithm. Note that it is almost identical to Algorithm \ref{alg:equivalence-linearithmic}, since we only substitute the marginal prefix oracle $\tau|_i^{x_{[i-1]} = w_{[i-1]}}$ with the marginal oracle of $\mu$.

\begin{algorithm}[H]
    \caption{$\eps$-test for binary-alphabet \propname{product}}
    \label{alg:binary-productness-linearithmic}
    \begin{algorithmic}
        \State \textbf{input} $n$, $\eps > 0$, a distribution $\mu$ over $\{0,1\}^n$.
        \State \phantom{\textbf{input}} $\mu$ is accessible through the prefix oracle.
        \State \textbf{completeness} If $\mu \in \propname{product}$, then the output is \accept with probability at least $\frac{2}{3}$.
        \State \textbf{soundness} If $\min\limits_{\tau \in \propname{product}} \!\!\dtv(\mu,\tau) \!>\! \eps$, then the output is \reject with probability at least $\frac{2}{3}$.
        \State \textbf{let} $\pi$ be the uniform distribution over $[n]$.
        \State \textbf{let} $Y$ be a random variable that distributes as $\pi\times\mu$.
        \State \textbf{let} $X$ be a random variable defined as a function of $Y=(i,w)$:
        \[X(i,w)=\chi^2\left(\mu|_i^{x_{[i-1]} = w_{[i-1]}}(1), \mu|_i(1)\right)\]
        \State \textbf{let} $\rho \gets \frac{\eps^2}{24 \log(\eps/2n)}$.
        \State \textbf{run} Levin's procedure (Algorithm \ref{alg:levin}), where its input consists of $Y$, $\rho/n$, and the black box $((i,w),\eps') \to (\text{Algorithm \ref{alg:chi-square-single} with input $\mu|_i^{x_{[i-1]} = w_{[i-1]}}, \mu|_i$ and $\eps'$})$.
        \If{Levin's procedure accepts}
            \State \Return \accept.
        \Else
            \State \Return \reject.
        \EndIf
    \end{algorithmic}
\end{algorithm}

\end{document}